\documentclass[runningheads]{llncs}

\usepackage[T1]{fontenc}
\usepackage[utf8]{inputenc}
\usepackage{amssymb, bm}
\usepackage{braket}
\usepackage{graphicx}
\usepackage{hyperref}
\usepackage{color}
\usepackage{physics}
\usepackage{booktabs}
\usepackage{enumitem}
\usepackage{cite}   

\newcommand{\CP}{\mathbb{CP}}
\newcommand{\C}{\mathbb{C}}
\newcommand{\R}{\mathbb{R}}
\newcommand{\Herm}{\mathrm{Herm}}
\newcommand{\SL}{\mathrm{SL}}
\newcommand{\SU}{\mathrm{SU}}
\newcommand{\SO}{\mathrm{SO}}

\newcommand{\bx}{\bm{x}}
\newcommand{\bsig}{\bm{\sigma}}
\newcommand{\bn}{\hat{\bm{n}}}
\newcommand{\dg}{\dagger}
\newcommand{\inner}[2]{\langle #1 \vert #2 \rangle}
\newcommand{\arxiv}[1]{\href{https://arxiv.org/abs/#1}{arXiv:#1}}

\begin{document}

\title{The Hermitian inner product selects the time axis,\\ the Born rule measures it}

\titlerunning{The inner product selects the time axis}

\author{Sebastian Zaj\k{a}c}

\institute{Military University of Technology, Warsaw, Poland\\
\email{sebastian.zajac@wat.edu.pl}}

\maketitle

\begin{abstract}
\noindent
The correspondence between $2\times 2$ Hermitian matrices and Minkowski $4$-vectors
recovers Lorentzian symmetries from the internal degrees of freedom of a qubit, with no
reference to an external spacetime. Recent work characterises the resulting Lorentz
invariants and leaves the \emph{mechanism} of emergence---what singles out a time
direction---as an explicit open question. We give an elementary answer and, in doing so,
correct a natural misattribution. The bare spin space $(\C^2,\varepsilon)$ is
$\SL(2,\C)$-symmetric and singles out no axis; so is the null cone it generates. What
selects a future-timelike axis is the choice of a Hermitian inner product, equivalently
a positive reference form $\sigma^0$: this choice---made in passing from a normed space
to a Hilbert space, \emph{before} any probability is assigned---reduces $\SL(2,\C)$ to
its maximal compact $\SU(2)$, the stabiliser of $\sigma^0$. The Born rule enters one
level later: $\inner{\xi}{\xi}=\tr(\sigma^0\,\xi\xi^\dg)$ is the projection of the
state's null vector onto $\sigma^0$, i.e.\ its energy in that frame, and under a boost
it rescales as a Doppler shift. Thus the Hilbert structure selects the axis; the Born
rule is where that axis becomes a measurable energy and where the frame-dependence of
$\lvert\psi\rvert^2$ becomes empirical. The ingredients are classical; what we add is
their identification as the mechanism the recent programme leaves open, with the
symmetry-breaking step located precisely. This is a kinematic identification of that
step, not a dynamical account of why a particular axis is selected. We close by handing
back the many-qubit case, where the datum is a tuple of such choices.
\end{abstract}

\section{Introduction}

Fullwood, Vedral and Guzm\'an-Gonz\'alez have recently argued that Lorentzian
symmetries are \emph{internal} to quantum information: the restricted Lorentz group is
fixed by demanding preservation of purity, and spectral invariants such as the linear
mutual information appear as Lorentz invariants of qubit ensembles
\cite{FVG2026,FV2025}. That programme characterises what is \emph{invariant}, and its
authors are explicit that turning the correspondence into a concrete \emph{mechanism}
of emergence---in particular, what singles out a time direction---remains open. This
note offers one building block toward that mechanism, drawn from the elementary
structure the correspondence already carries.

It is tempting to say that the Born rule---the act of forming a probability
$\lvert\psi\rvert^2$---is what breaks Lorentz covariance and picks out a time direction.
That is not quite right, and locating the step correctly is part of the point. The
breaking happens one level earlier. A complex vector space carries no distinguished
axis until it is equipped with a Hermitian inner product; that is the step from a mere
normed (Banach) space to a Hilbert space, and it is logically prior to assigning any
probabilities. As we show, it is exactly this choice---of the inner product, equivalently
of a positive reference form $\sigma^0$---that reduces $\SL(2,\C)$ to $\SU(2)$ and
selects a future-timelike axis. The Born rule then operates on an \emph{already chosen}
$\sigma^0$; its role is not to break the symmetry but to \emph{measure} the break---to
turn the abstract axis into a number, the energy of the state's null vector in the
chosen frame. The one-line summary is therefore: \emph{the Hilbert structure selects the
time axis, and the Born rule reads it off as an energy.}

We fix the scope plainly, since it bounds what follows. This is a \emph{kinematic}
identification: it names the datum whose choice breaks the symmetry, and shows that
datum to be the inner product rather than the probability rule. It is not a
\emph{dynamical} account. We do not propose a process by which a particular $\sigma^0$
comes to be selected, nor do we claim the observation as a fundamental theory of time;
the datum is put in by hand, through the Hilbert-space structure, and our contribution
is only to locate it. Why a definite $\sigma^0$ should be singled out---in what physical
process, if any---is a separate and harder question, which we flag as the natural
continuation rather than address here.

We stress at the outset that the ingredients are elementary and, individually, standard:
the identification of $\xi^\dg\xi$ with the time component of a spinor's null
``flagpole'' is in \cite{PenroseRindler}; the little group of a timelike vector is
$\SU(2)$; and the frame-dependence of reduced quantum states under boosts is known from
relativistic quantum information \cite{PST2002}. The contribution is their assembly into
a mechanism, with the symmetry-breaking datum named precisely.

\section{The bare space and the null cone are axis-free}

Let $S=\C^2$ be the spin space, with vectors $\xi=(a,b)^{\!\top}$. On $S$ the only
$\SL(2,\C)$-invariant bilinear is the antisymmetric symplectic form
$\varepsilon(\xi,\eta)=\xi^{\!\top}\!\varepsilon\,\eta$,
$\varepsilon=\big(\begin{smallmatrix}0&1\\-1&0\end{smallmatrix}\big)$, and it is
\emph{signature-free}: from an antisymmetric form no metric is built. To reach Minkowski
signature one pairs $S$ with its complex conjugate $\bar S$. Fixing
$\sigma^\mu=(\mathbb 1,\sigma^1,\sigma^2,\sigma^3)$, the real span of Hermitian matrices
is $\Herm(2)=S\otimes\bar S$ under
\begin{equation}\label{eq:map}
  \R^{1,3}\xrightarrow{\ \sim\ }\Herm(2),\qquad
  x\mapsto X=x^0\mathbb 1+\bx\cdot\bsig,\qquad
  \det X=\eta_{\mu\nu}x^\mu x^\nu,\quad \tr X=2x^0,
\end{equation}
with $\eta=\mathrm{diag}(+,-,-,-)$; the group acts by $X\mapsto LXL^\dg$, covering
$\SO^+(1,3)$. This dictionary, and its reading as a link between qubit operations and
special relativity, is classical \cite{PenroseRindler,ArrighiPatricot}.

\begin{lemma}[Null map]\label{lem:null}
For $\xi\neq 0$ the matrix $X(\xi)=\xi\xi^\dg\in\Herm(2)$ is positive semidefinite of
rank one, so its $4$-vector is future-pointing and null. Its components are
\begin{equation}\label{eq:comp}
  x^0=\tfrac12\big(|a|^2+|b|^2\big)=\tfrac12\,\xi^\dg\xi,\qquad
  \bx=\tfrac12\,\xi^\dg\bsig\,\xi ,
\end{equation}
the map is invariant under $\xi\mapsto e^{i\theta}\xi$, and it descends to the density
operator $\rho=X/\tr X=\tfrac12(\mathbb 1+\bn\cdot\bsig)$, $\bn=\bx/x^0$. Thus
$\xi\mapsto X(\xi)$ is the passage $\ket\psi\mapsto\ket\psi\bra\psi$, and $\bn$ is the
point of $\CP^1$ (the Bloch, equivalently celestial, sphere) carried by the ray $[\xi]$.
\end{lemma}

\begin{proof}
$X=\xi\xi^\dg$ has eigenvalues $\{\xi^\dg\xi,0\}$, so is PSD with $\det X=0$; by
\eqref{eq:map} the vector is null and $x^0=\tfrac12\tr X>0$. Components follow from
$x^0=\tfrac12\tr X$, $x^i=\tfrac12\tr(\sigma^iX)$. Under $\xi\mapsto e^{i\theta}\xi$,
$X\mapsto X$, so $X$ depends only on $[\xi]$ and $\tr X$; normalising gives $\rho$.
\end{proof}

Two things are worth noting. The overall phase is the fibre of the Hopf map $S^3\to S^2$,
discarded by the map. And forming $\xi\xi^\dg$ already uses the dagger, i.e.\ the
antilinear identification $S\to\bar S$: the Hermitian-matrix picture presupposes a
reality structure. What it does \emph{not} yet presuppose is any preferred timelike
vector: the null cone is $\SL(2,\C)$-homogeneous, $\SL(2,\C)$ acts on $\CP^1$ by
conformal (M\"obius) maps, and nothing so far singles out an axis. Everything
Lorentz-covariant to this point is blind to a time direction---as are all the invariants
of \cite{FVG2026,FV2025}: linear entropy, mutual information, concurrence, the singlet
correlation function are $\SL(2,\C)$-scalars, hence blind to phase and to normalisation
alike.

\section{The choice of inner product selects the axis}

The question ``what selects a time direction?'' is therefore ``what breaks $\SL(2,\C)$
to a subgroup fixing an axis?''. The answer is the datum one must add to turn $S$ into a
Hilbert space: a Hermitian inner product $\inner{\xi}{\eta}=\xi^\dg\sigma^0\eta$,
equivalently the choice of the positive reference form $\sigma^0=\mathbb 1$.

\begin{proposition}[The inner product is the time axis]\label{prop:main}
The Hermitian inner product on $S$ is not $\SL(2,\C)$-invariant, and
\begin{equation}\label{eq:stab}
  \inner{L\xi}{L\eta}=\inner{\xi}{\eta}\ \ \forall\,\xi,\eta
  \quad\Longleftrightarrow\quad L^\dg L=\mathbb 1
  \quad\Longleftrightarrow\quad L\in\SU(2).
\end{equation}
Under the vector action $L:\sigma^0\mapsto L\sigma^0L^\dg$ on $\Herm(2)$, the stabiliser
of the future-timelike vector $\sigma^0=\mathbb 1$ is exactly this same $\SU(2)$. Hence
choosing the inner product $\sigma^0$ is choosing a unit future-timelike axis, and the
group preserving that choice is the rotations.
\end{proposition}

\begin{proof}
$\xi^\dg L^\dg L\,\eta=\xi^\dg\eta$ for all $\xi,\eta$ iff $L^\dg L=\mathbb 1$, i.e.\ $L$
unitary; with $\det L=1$ this is $L\in\SU(2)$. The little group of $\sigma^0=\mathbb 1$
under $X\mapsto LXL^\dg$ is $\{L:LL^\dg=\mathbb 1\}=\SU(2)$, the same subgroup.
\end{proof}

The two clauses say the same thing from opposite sides: the automorphisms of the Hilbert
inner product and the stabiliser of the timelike axis $\sigma^0$ coincide. $\SU(2)$ acts
on $\CP^1$ by isometries of the round sphere (it fixes the axis); the boosts
$\SL(2,\C)/\SU(2)$ move the axis and act by non-isometric conformal maps. The
symmetry-breaking step is the passage from $(S,\varepsilon)$ to $(S,\sigma^0)$---from a
space with only a symplectic form to a Hilbert space with a positive form---and it is
complete before any state is normalised or any probability computed.

\section{The Born rule measures it: energy and frame-dependence}

Given $\sigma^0$, the Born rule is the evaluation of that form on the state,
\begin{equation}\label{eq:proj}
  \inner{\xi}{\xi}=\xi^\dg\sigma^0\xi=\tr\!\big(\sigma^0\,X(\xi)\big)=2x^0=2\,\eta(t,x),
  \qquad t\leftrightarrow\sigma^0 .
\end{equation}
It does two things. It \emph{reads off} the $\sigma^0$-component of the state's null
vector---the datum already fixed by the choice of form. And it \emph{assigns a physical
value} to that component: $x^0$ is the energy (frequency) of the null datum $\xi\xi^\dg$
in the frame $\sigma^0$. This is where the abstract axis-choice of
Proposition~\ref{prop:main} becomes measurable.

\begin{corollary}[Doppler; the norm is a frequency]\label{cor:doppler}
Under a boost $L=\exp(\tfrac{\phi}{2}\,\hat{\bm m}\cdot\bsig)$ (Hermitian,
$L^\dg=L\neq L^{-1}$), i.e.\ a change of $\sigma^0$, the Born norm transforms as
$\inner{\xi}{\xi}\mapsto\inner{\xi}{L^\dg L\,\xi}$, a rescaling; since
$\inner{\xi}{\xi}=2x^0$, this is exactly the redshift of the frequency of $\xi\xi^\dg$.
Different observers assign different values to $\lvert\psi\rvert^2$ because they project
the same null ray onto different time axes; rotations $L\in\SU(2)$ leave it invariant.
This is the $2$-spinor face of the frame-dependence of reduced quantum descriptions
\cite{PST2002}.
\end{corollary}

So the division of labour is clean. The Hilbert structure ($\sigma^0$) breaks
$\SL(2,\C)$ and \emph{selects} the time axis; the Born rule \emph{measures} it, turning
the axis into an energy and making the frame-dependence of $\lvert\psi\rvert^2$ an
empirical statement rather than a formal one. Neither step is the other: a bare density
operator carries $\sigma^0$ implicitly (it is normalised), while the Born number exposes
what that implicit choice costs under a boost.

\begin{remark}[Causal cone $=$ positivity cone]
The map \eqref{eq:map} carries the causal structure onto operator positivity, once
$\sigma^0$ is fixed: future timelike $\Leftrightarrow$ positive definite (a faithful
density matrix after normalisation), future null $\Leftrightarrow$ rank-one PSD (a pure
state $\xi\xi^\dg$), spacelike $\Leftrightarrow$ indefinite (no quantum state). The
closed future cone is the symmetric positivity cone of the Jordan algebra $\Herm(2)$;
its extreme rays are the pure states, the image of Lemma~\ref{lem:null}.
\end{remark}

\begin{remark}[Where the single sign lives]
Splitting $\Herm(2)=\R\mathbb 1\oplus\mathfrak{su}(2)$ gives
$\det X=(x^0)^2-\lvert\bx\rvert^2$: the one timelike sign is the trace
($\mathfrak u(1)$) part, the three spacelike signs the traceless ($\mathfrak{su}(2)$)
part. The distinguished axis is the same $\mathbb 1=\sigma^0$ whose choice is the
inner product; the signature $(1,3)$ is the split
$\mathfrak u(2)=\mathfrak u(1)\oplus\mathfrak{su}(2)$, not the antisymmetric form
$\varepsilon$, which is signature-free.
\end{remark}

\section{Outlook and a question}

The statement is deliberately narrow: it names the datum that plays the role of the
emergence mechanism left open in \cite{FVG2026}, and locates it one step before the
probability rule. We close with the question this raises for that programme. For a single
qubit the Hilbert structure is one choice of $\sigma^0$, hence one timelike axis. For an
$n$-qubit ensemble the invariants of \cite{FVG2026} (the spectrum of $W=\rho\rho^\star$,
and $\Tr W=I_L$) are $\SL(2,\C)^{\otimes n}$ invariants, i.e.\ each factor may be
transformed independently. The complementary datum is then not a single inner product
but a \emph{tuple}---one Hilbert structure, one $\sigma^0$, per factor, hence one time
axis per subsystem. Two readings seem possible: either a single global $\sigma^0$ is
selected (a shared rest frame, breaking $\SL(2,\C)^{\otimes n}$ to a diagonal $\SU(2)$),
or each subsystem carries its own axis and the ``relative time'' between subsystems
becomes an invariant built from the mismatch. Which of these is forced---and whether the
second is related to the temporal or pseudo-entropic structures discussed in
emergent-spacetime approaches---is the natural next question, and one we would be glad to
hear the authors' view on.

Beyond the multi-qubit case, the deeper open direction is the one the scope remark set
aside: a \emph{dynamical} account in which $\sigma^0$ is not posited but selected by some
process---for instance where a family of states carries its own structure in place of a
fixed inner product. Identifying the datum, as done here, is the prerequisite for asking
that question sharply; answering it is a separate undertaking.

\end{document}